\documentclass[aps,  amsmath,  amssymb,  amsfonts,  showpace,  twocolumn,  pra]{revtex4}

\usepackage{amscd}
\usepackage{amsthm}
\usepackage{graphicx}
\usepackage{epstopdf}
\usepackage{mathdots}
\usepackage{blkarray}
\makeatletter
\newbox\BA@first@box
\makeatother

\newtheorem{Thm}{Theorem}
\newtheorem{Lem}{Lemma}

\newtheorem{Prop}{Proposition}

\usepackage{ulem}

\usepackage[
  colorlinks,
  linkcolor = blue,
  citecolor = blue,
  urlcolor = blue]{hyperref}

\usepackage{amsmath,  amssymb,  amsthm}

\usepackage{tikz}
\usetikzlibrary{arrows}
\usetikzlibrary{calc,  positioning}

\usepackage{thm-restate}

\usepackage{colonequals}

\usepackage[font=small]{caption}

\newtheoremstyle{noCaption}
{\topsep}
{\topsep}
{\itshape}
{}
{}
{}
{0pt}
{}%

\begin{document}
\title{Distinguishability-based genuine nonlocality with genuine multipartite entanglement}
\author{Zong-Xing Xiong$^{1}$\footnote{hongbe123@outlook.com}, Mao-Sheng Li$^{2}$, Zhu-Jun Zheng$^{2}$, Lvzhou Li$^{1}$\footnote{lilvzh@mail.sysu.edu.cn}\\
{\footnotesize{\it $^1$Institute of Quantum Computing and Computer Theory, School of Computer Science and Engineering,}}\\
{\footnotesize{\it Sun Yat-Sen University, Guangzhou 510006, China}}\\
{\footnotesize{\it $^2$School of Mathematics, South China University of Technology, Guangzhou 510640, China}}\\
}
\affiliation{
}

\begin{abstract}
A set of orthogonal multipartite quantum states is said to be distinguishability-based genuinely nonlocal (also genuinely nonlocal, for abbreviation) if the states are locally indistinguishable across any bipartition of the subsystems. This form of multipartite nonlocality, although more naturally arising than the recently popular ``strong nonlocality'' in the context of local distinguishability, receives much less attention. In this work, we study the distinguishability-based genuine nonlocality of a typical type of genuine multipartite entangled states --- the $d$-dimensional GHZ states, featuring systems with local dimension not limited to $2$. In the three-partite case, we find the existence of small genuinely nonlocal sets consisting of these states: we show that the cardinality can at least scale down to linear in the local dimension $d$, with the linear factor $l=1$. Specifically, the method we use is semidefinite program and the GHZ states to construct these sets are special ones which we call ``GHZ-lattices''. This result might arguably suggest a significant gap between the strength of strong nonlocality and the distinguishability-based genuine nonlocality.  Moreover, we put forward the notion of $(s,n)$-threshold distinguishability and utilizing a similar method, we successfully construct $(2,3)$-threshold sets consisting of GHZ states in three-partite systems.
\end{abstract}
\pacs{03.  67.  Hk,  03.  65.  Ud }
\maketitle

\section{Introduction}
Quantum nonlocality is one of the most surprising properties in quantum mechanics. The most well-known manifestation of nonlocality --- Bell nonlocality \cite{Bell1964,Clauser1969,Brunner2014}, which is revealed by violation of Bell type inequilities \cite{Clauser1969,Freedman1972,Aspect1981,Aspect1982,Hensen2015,BIG Bell Test Collaboration}, can only arise from entangled states. Apart from Bell nonlocality, there are also other forms of nonlocality. Among them, the distinguishability-based nonlocality, which concerns the problem of locally distinguishing a certain set of orthogonal quantum states, has attracted much attention. It serves to explore the fundamental question about locally accessing global information and also the relation between entanglement and locality, which are of central interest in quantum information theory. Unlike Bell nonlocality however, entanglement is not necessary for this kind of nonlocality. In \cite{Bennett99}, Bennett et al. presented the first example of orthogonal product states that are indistinguishable when only local operations and classical communications are allowed, which was known as ``quantum nonlocality without entanglement''. From then on, such nonlocality based on distinguishability has been studied extensively (see \cite{Bennett99,Bennett1999,Walgate00,Walgate02,DiVicenzo03,Horodecki03,Hayashi06,Bandyopadhyay09,Bandyopadhyay10,Niset06,Ghosh01,Ghosh04,Fan04,Nathanson05,Bandyopadhyay11,Tian2015,Tian15,Xiong19,Li2019,Yu12,Nathanson13,Cosentino13,CosentinoR14,Li15,Yu15,Xu2016,Zhang2016,Zhang2017,Halder2018,Li2018,Jiang2020} for an incomplete list). Moreover, the local discrimination of quantum states has been practically applied in a number of distributed quantum protocols such as quantum data hiding \cite{Terhal2001,DiVicenzo2002,Eggeling2002,Matthews2009} and quantum secret sharing \cite{Markham2008,Hillery1999,Schauer2010,Cleve1999,Gottesman2000,Rahaman2015,Wang2017}.

Although not necessary for distinguishability-based nonlocality, evidences have been found showing that entanglement can somehow raise difficulty in quantum state discrimination in many occasions. To what extent entanglement is responsible for indistinguishability has been one of the main focuses \cite{Horodecki03,Hayashi06,Bandyopadhyay09,Bandyopadhyay10}. In \cite{Hayashi06}, Hayashi et al. observed that the number of pure states that can be perfectly distinguished locally is bounded above by the total dimension over the average entanglement of the states. This result quantitatively shows that states with more entanglement will generally (while not always) be more difficult to be distinguished. Meanwhile, lots of studies had focused on local distinguishability of the maximally entangled states \cite{Ghosh01,Ghosh04,Fan04,Tian2015,Tian15,Xiong19,Li2019}, as opposed to studying the distinguishability of product states. Notably, it can be deduced directly from the results of Hayashi et al. \cite{Hayashi06} that any $k > d$ orthogonal maximally entangled states in $\mathbb{C}^d \otimes \mathbb{C}^d$ are not locally distinguishable. This fact was actually first revealed by Nathanson who also showed that any three orthogonal maximally entangled states in $\mathbb{C}^3 \otimes \mathbb{C}^3$ is locally distinguishable \cite{Nathanson05}. In spite of these quantitative results indicating the maximal number of perfectly distinguishable states, we have on the other hand little idea on the minimal number of locally indistinguishable states. What is known to us is that any two orthogonal pure states can always be locally distinguished, no matter whether the states are entangled or not \cite{Walgate00}. Frustratedly, whether there exist three locally indistinguishable maximally entangled states in system $\mathbb{C}^d \otimes \mathbb{C}^d$ for $d > 3$ remains unknown today. Therefore, efforts had been paid in seeking sets of $k \leq d$ orthogonal maximally entangled states that are not perfectly distinguishable by local operations and classical communications (LOCC) \cite{Nathanson05,Bandyopadhyay11,Yu12,Nathanson13,Cosentino13,CosentinoR14,Li15,Yu15}. The best result so far is due to Yu and Oh \cite{Yu15}, who showed that the cardinality of indistinguishable sets of maximally entangled states in $\mathbb{C}^d \otimes \mathbb{C}^d$ can asymptotically scale down to $3d/4$ when $d$ is large.

Not long ago, Halder et al. \cite{Halder19} introduced the concept of local irreduciblity, which is a stronger form of local indistinguishability. A set of orthogonal multipartite quantum states is said to be locally irreducible if one cannot eleminate one or more states from the whole set, with the restriction that only orthogonality-preserving local measurements are allowed. They revealed the phenomenon of ``strong nonlocality without entanglement'', which is a nontrivial
generation of Bennett's nonlocality without entanglement \cite{Bennett99}. Strong nonlocality of a set of orthogonal multipartite quantum states refers to local irreducibility of these states through every bipartition of the subsystems and the authors of \cite{Halder19} presented the first examples of strong nonlocal sets of product states in $\mathbb{C}^3 \otimes \mathbb{C}^3 \otimes \mathbb{C}^3$ and $\mathbb{C}^4 \otimes \mathbb{C}^4 \otimes \mathbb{C}^4$.  After that, much attention has been paid to study the strong nonlocality of multipartite quantum states \cite{Zhang19,Shi20,Gao,Gao2,Yuan20,Wang21,Shi22,Shi2022,Li22arxiv,Bandyopadhyay21,Li22}. In \cite{Yuan20}, the authors constructed $6(d-1)^2$ orthogonal product states in three-partite system $\mathbb{C}^d \otimes \mathbb{C}^d \otimes \mathbb{C}^d$ that is strongly nonlocal. Soon after, Wang et al. successfully constructed strongly nonlocal sets with genuine multipartite entanglement \cite{Wang21}. The strongly nonlocal sets they constructed are consisted of GHZ states in three-partite system $\mathbb{C}^d \otimes \mathbb{C}^d \otimes \mathbb{C}^d$ ($d \geq 3$) and have cardinality $d^3 - (d-2)^3$ ($d^3 - (d-2)^3 + 2$) when $d$ is odd (resp. even). Shi et al. gave more general results in $(\mathbb{C}^d)^{\otimes N} (N \geq 3)$, showing the existence of strongly nonlocal orthogonal entangled sets with size $d^N - (d-1)^N +1$ \cite{Shi22}. They also constructed $18$ strongly nonlocal genuine multipartite entangled states in $\mathbb{C}^3 \otimes \mathbb{C}^3 \otimes \mathbb{C}^3$, outperforming the former result of $19$ strongly nonlocal UPB in the same tripartite system \cite{Shi2022}.

In \cite{Rout19}, Rout et al. first studied the concept of distinguishability-based genuine nonlocality (there the authors named it genuine nonlocality simply; we also adopt this abbreviation here sometimes when no ambiguity occurs): a set of orthogonal multipartite quantum states is said to be (distinguishability-based) genuinely nonlocal if the states are locally indistinguishable across any bipartition of the subsystems. Further in \cite{Li21,Rout21}, the authors constructed sets of multipartite product states which are genuinely nonlocal. Namely, they found genuine nonlocality without entanglement. This concept of multipartite nonlocality is weaker but more natural than strong nonlocality in the context of distinguishability. It's obvious that strong nonlocality always implies genuine nonlocality \cite{Halder19}, yet in what extent the former is stronger than the latter, little is known.

In this paper, we are to study the distinguishability-based genuine nonlocality of a typical type of genuine multipartite entangled states --- the ($d$-dimensional) GHZ states on qudit systems. More explicitly, we focus on the three-partite case, where three-qudit GHZ states are maximally entangled in the sense that they are maximally entangled in all bipartitions. Similar to the bipartite case when maximally entangled states are studied, one can derive from the result of  Hayashi et al. \cite{Hayashi06} that any $k > d^2$ three-qudit GHZ states are genuinely nonlocal. However, we have yet no idea about the minimal number of such states having the property of genuine nonlocality. Here, we make a step forward by showing the existence of certain genuinely nonlocal sets of GHZ states that have rather small cardinality. We find that the cardinality can at least scale down to $d+3$, for the cases when local dimension $d$ are powers of $2$. The method we use here is inspired by  Cosentino \cite{Cosentino13,CosentinoR14} and the special kind of GHZ states to construct these sets are refered to as ``GHZ-lattices'' by us. Up till now, since no existing strongly nonlocal sets that've been constructed has such small magnitude, it's reasonable for us to argue that there might exist a substantial difference between the strength of strong nonlocality and the distinguishability-based genuine nonlocality. Furthermore, we put forward the notion of $(s,n)$-threshold distinguishability in $n$-partite system as an extension of genuine nonlocality. Adopting the similar method of semidefinite program in our former discussion,  we successfully construct $(2,3)$-threshold sets consisting of GHZ states in three-partite systems.

The rest of this paper is organized as follows: In Sec. \ref{section2}, we present some necessary notations and definitions. In Sec. \ref{section3}, we study the genuine nonlocality of the GHZ states in three-partite system $\mathbb{C}^d \otimes \mathbb{C}^d \otimes \mathbb{C}^d$. Taking advantage of the structure of ``GHZ-lattices'' and also the method of semidefinite program, we construct genuinely nonlocal sets of these states with conspicuously small cardinality. In Sec. \ref{section4}, we extend the notion of genuine nonlocality to the more general $(s,n)$-threshold distinguishability in multipartite systems and we construct $(2,3)$-threshold sets consisting of GHZ-lattices in three-partite systems. Finally, we draw our conclusion and present some relating problems in Sec. \ref{section5}.

\section{Preliminaries}\label{section2}
\textit{Genuine multipartite entanglement}:
In bipartite system $\mathcal{H}_A \otimes \mathcal{H}_B$, a pure state $|\Psi\rangle_{AB}$ is called entangled if it cannot be written as the tensor product of two pure states of the two subsystems. Namely, it's not of the form $|\Psi\rangle_{AB} = |\alpha\rangle_A \otimes |\beta\rangle_B$. For pure state $|\Psi\rangle_{A_1 \cdots A_n}$ in multipartite system $\mathcal{H}_{A_1} \otimes \cdots \otimes \mathcal{H}_{A_n}$, it is called genuinely multipartite entangled if it is entangled for any bipartition of the subsystems $\{A_1, \cdots, A_n\}$. The most well known kind of genuinely (three-partite) entangled states are the GHZ states and the W states in three-qubit system $\mathbb{C}^2 \otimes \mathbb{C}^2 \otimes\mathbb{C}^2$. They are in fact the only two types of genuine tripartite entangled states in $\mathbb{C}^2 \otimes \mathbb{C}^2 \otimes\mathbb{C}^2$ up to SLOCC equivalence \cite{Dur2000,Acin2001}.

\textit{$d$-dimensional GHZ (Greenberger-Horne-Zeilinger) states} \cite{ghz1990}:
In $n$-partite system $\mathbb{C}^{d} \otimes \mathbb{C}^{d} \otimes \cdots \otimes \mathbb{C}^{d}$, the states of the form
\begin{equation}
\frac{1}{\sqrt{d}} \sum_{k=0}^{d-1} |\xi_k^{(1)} \xi_k^{(2)} \cdots \xi_k^{(n)}\rangle
\end{equation}
are called $d$-dimensional GHZ states, where each $\{|\xi_k^{(i)}\rangle\}_{k=0}^{d-1}$ is any set of orthogonal basis for the $i$th subsystem. These states can be proved to be genuinely entangled and they have been considered as generic resources in many quantum information processing tasks \cite{Markham2008,Hillery1999,Schauer2010,Cleve1999,Gottesman2000,Rahaman2015,Wang2017,Karlsson1998,Shi2000,Joo2003}.

\textit{Local distinguishability}:
A set of orthogonal (pure) multipartite quantum states, which is priorly known to several spatially separated parties, is said to be locally distinguishable if the parties are able to tell exactly which state they share through some protocol, provided only local measurements and classical communications are allowed. In the literature, since the mathematical structure of LOCC measurement is rather complicated, it is usually the separable measurement or the positive-partial-transpose (PPT) measurement that is considered \cite{Bandyopadhyay13,Bandyopadhyay15,Duan09,Yu14,Li17}.

\textit{Local irreducibility}:
A set of multipartite orthogonal quantum states is said to be locally irreducible if it is not possible to locally eliminate one or more states from the set while preserving orthogonality of the postmeasurement states. Typical examples of locally irreducible sets include the two-qubit Bell basis and the three-qubit GHZ basis \cite{Halder19}.  While a locally irreducible set must be locally indistinguishable, the opposite is generally not true.

\textit{Distinguishability-based genuine nonlocality} and \textit{strong nonlocality}:
\ \ A set of orthogonal multipartite quantum states is called (distinguishability-based) genuinely nonlocal if the states are locally indistinguishable across every bipartition of the subsystems. Furthermore, if the states are locally irreducible across every bipartition, then they are called strongly nonlocal. Straightforwardly, a strongly nonlocal set must be also genuinely nonlocal, while the converse is not true. As it is shown immediately in the next section, while the three-qubit GHZ basis (\ref{GHZ}) failed to be strongly nonlocal \cite{Halder19}, they are genuinely nonlocal obviously.

\section{Genuine nonlocality for the three-qudit GHZ states}\label{section3}
In three-qudit system $\mathbb{C}^d \otimes \mathbb{C}^d \otimes \mathbb{C}^d$, the $d$-dimensional GHZ states are maximally entangled in the sence that all their reductions to $\lfloor \frac32 \rfloor$-qudit are maximally mixed. Here in this section, we are to discuss the genuine nonlocality of this typical kind of genuine three-partite entangled states. Roughly speaking, a set of more states might often appear harder for distinguishing while a set with less states is usually more likely to be distinguishable. For example, all supersets of an indistinguishable set are always indistinguishable, while on the other hand, all subsets of a distinguishable one are certainly also distinguishable. Therefore, one is interested in both the maximal number of states that might be distinguishable and also the minimal number of states that are indistinguishable (in our case, genuinely nonlocal). The following lemma that can be derived from Hayashi et al. \cite{Hayashi06} shows that the number of $d$-dimensional GHZ states in $\mathbb{C}^d \otimes \mathbb{C}^d \otimes \mathbb{C}^d$ that are not genuinely nonlocal is at most $d^2$:

\begin{Lem}\label{lemma0}\textnormal{\cite{Hayashi06}}
In three-partite system $\mathbb{C}^d \otimes \mathbb{C}^d \otimes \mathbb{C}^d$, any set of $s > d^2$  orthogonal $d$-dimensional GHZ states is genuinely nonlocal.
\end{Lem}

Here however, we will mainly focus on seeking small number of these  states that are genuinely nonlocal. We first discuss the simplest situation when $d=2$, namely the three-qubit case. In $\mathbb{C}^2 \otimes \mathbb{C}^2 \otimes\mathbb{C}^2$, the following $8$ orthogonal states:
\begin{equation}\begin{split} \label{GHZ}
|\psi_{0,7}\rangle=\frac{|000\rangle \pm |111\rangle}{\sqrt{2}}, \ \ |\psi_{1,6}\rangle=\frac{|001\rangle \pm |110\rangle}{\sqrt{2}}, \\
|\psi_{2,5}\rangle=\frac{|010\rangle \pm |101\rangle}{\sqrt{2}}, \ \ |\psi_{3,4}\rangle=\frac{|011\rangle \pm |100\rangle}{\sqrt{2}}, \\
\end{split}\end{equation}
consist a set of orthogonal basis. Here, we place the conjugate pairs together and assign them subscript indices which sum up to $2^3-1$. They are refered to as the three-qubit GHZ basis and are important in a number of quantum information processing scenarios.
In \cite{Halder19}, the authors showed that these states are not strongly nonlocal. Namely, when certain two of the three parties are allowed to join together, it is possible to locally eliminate one or more states from the whole set while preserving orthogonality of the postmeasurement states. Here, instead of discussing about the reducibility, we focus on the bipartite distinguishability problem of the three-qubit GHZ basis. It's immediate from Lemma \ref{lemma0} that: {\normalem \em Any $s \geq 5$ states of the three-qubit GHZ basis (\ref{GHZ}) are distinguishability-based genuinely nonlocal}. For the other side, results of \cite{Bandyopadhyay10} showed that this bound is tight, in the sense that there exist $4$ states among the three-qubit GHZ basis (\ref{GHZ}) that are not genuinely nonlocal. Here, we further show that any $4$ states among the three-qubit GHZ basis are not genuinely nonlocal. We give an elegant and visible way to demonstrate this fact.

\begin{Prop}\label{ghz4}
Any $4$ states of the three-qubit GHZ basis (\ref{GHZ}) can be locally distinguished by at least one of the bipartition $A|BC$, $B|CA$ or $C|AB$. That is, they are not genuinely nonlocal.
\end{Prop}

\begin{proof}
We demonstrate this fact with the aid of a schematic picture shown by FIG \ref{fig}. The four conjugate pairs are represented by four vertices of the tetrahedron and the coloured edges represent different bipartitions. Obviously, the $4$ states to be distinguished must be located on no less than two vertices of the tetrahedron in FIG \ref{fig}. If the $4$ states are located on four vertices of the tetrahedron, they can obviously be distinguished across all the bipartitions. For example, in bipartition $A|BC$, the qubit holders $BC$ can first perform local measurement $\{P_1^{(BC)} = |00\rangle\langle00| + |11\rangle\langle11|$,  $P_2^{(BC)} = |01\rangle\langle01| + |10\rangle\langle10|\}$ to reduced the states into two disjoint sets, with each being a $2$-ary subset of $\{|\psi_0\rangle, |\psi_7\rangle, |\psi_3\rangle, |\psi_4\rangle\}$ and $\{|\psi_1\rangle, |\psi_6\rangle, |\psi_2\rangle, |\psi_5\rangle\}$ respectively. Since any two states are always locally distinguishable \cite{Walgate00}, the $4$ states can be locally distinguished across bipartition $A|BC$. The same holds for $B|CA$ and $C|AB$; If the $4$ states to be distinguished are distributed on three vertices, they must be located on a certain face of the tetrahedron. Without loss of generality, suppose that these states are $\{|\psi_0\rangle, |\psi_7\rangle, |\psi_1\rangle, |\psi_2\rangle\}$. In this case, only in bipartition $A|BC$ can they be distinguished, for in $B|CA$ ($C|AB$), the subset $\{|\psi_0\rangle, |\psi_7\rangle, |\psi_2\rangle\}$ (resp. $\{|\psi_0\rangle, |\psi_7\rangle, |\psi_1\rangle\}$) is locally equivalent to three states among the Bell basis, which are known to be locally indistinguishable; If the $4$ states to be distinguished are on two vertices, they must be located on a certain edge. As an example, for $\{|\psi_0\rangle, |\psi_7\rangle, |\psi_3\rangle, |\psi_4\rangle\}$, they are not locally distinguishable through $A|BC$ while they can be distinguished through $B|CA$ and $C|AB$. To sum up, any $4$ states of the three-qubit GHZ basis (\ref{GHZ}) can be distinguished through at least one of the three bipartitions, i.e., they are not genuinely nonlocal.
\end{proof}

\begin{figure}[t]
\centering
\includegraphics[scale=0.45]{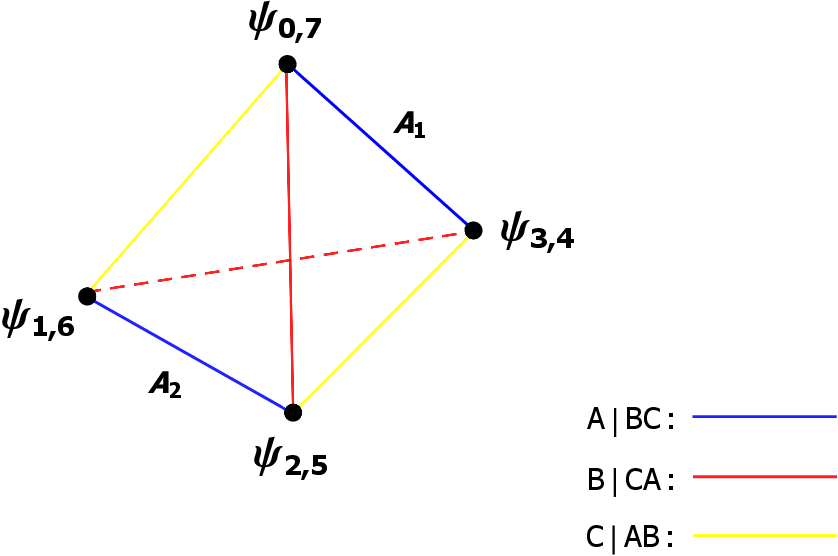}
\captionsetup{justification=raggedright}
\caption
{
A 3D schematic picture for local distinguishability of the three-qubit GHZ basis (\ref{GHZ}) in all bipartitions: \ Each vertex of the tetrahedron represents one conjugate pair of the basis. Four states on one edge are locally equivalent to the two-qubit Bell basis $\{|00\rangle \pm |11\rangle$,$|01\rangle \pm |10\rangle\}$, within the bipartition shown by the colour of that edge. \ For instance, in bipartition $A|BC$, the $4$ states of $\mathcal{A}_1$ are locally equivalent to the Bell basis and the same holds for $\mathcal{A}_2$. These two sets lie on the opposite edge of the tetrahedron colored blue, and states from different sets can be told apart if $BC$ perform 2- outcome joint measurement $\{P_1^{(BC)} = |00\rangle\langle00| + |11\rangle\langle11|,$ $P_2^{(BC)} = |01\rangle\langle01| + |10\rangle\langle10|\}$.
}\label{fig}
\end{figure}

Despite the nonexistence of genuinely nonlocal subsets of the three-qubit GHZ basis with cardinality $s < 2^2+1$, in systems of higher local dimension, we are using them to construct genuinely nonlocal sets of three-partite GHZ states with small cardinality --- much smaller than the trivial cardinality $d^2+1$ (indicated by Lemma \ref{lemma0}). What we are to discuss here are the special cases where the local dimension $d=2^t$, and the GHZ states are ``lattice-type'' ones, which are tensor products of the three-qubit GHZ basis (\ref{GHZ}). The idea here is primarily inspired by \cite{Cosentino13}, where small indistinguishable sets in bipartite systems had been construced using lattice-type maximally entangled states. Let $\overrightarrow{v} = (v_1, \cdots, v_t) \in \{0, 1, \cdots, 7\}^t$ be a $t$-element vector $(t \geq 1)$. Then the $t$-level \textit{lattice-type GHZ state} with index $\overrightarrow{v}$ in three-partite system $\mathbb{C}^{2^t} \otimes \mathbb{C}^{2^t} \otimes \mathbb{C}^{2^t}$ is defined by $|\Psi_{\overrightarrow{v}}\rangle = |\psi_{v_1}\rangle \otimes \cdots \otimes |\psi_{v_t}\rangle$. One can easily check that these $8^t$ states are $2^t$-dimensional GHZ states and they consist a set of orthogonal basis in $\mathbb{C}^{2^t} \otimes \mathbb{C}^{2^t} \otimes \mathbb{C}^{2^t}$. For abbreviation, here and after we also call this form of GHZ states ``GHZ-lattices''. It turns out when the multipartite states to be distinguished are all GHZ-lattices, the special structure of these states enable us to discuss their genuine nonlocality with much less efforts.

Since LOCC distinguishability is difficult to tackle mathematically, here instead we study the PPT-distinguishability of the multipartite states, for different bipartition respectively.  For bipartite system $\mathcal{H}_A \otimes \mathcal{H}_B$, if a set of orthogonal states is locally distinguishable, the states must be also PPT-distinguishable. To be more precise, the orthogonal bipartite states $\{|\phi_1\rangle,   \cdots,   |\phi_s\rangle\}$ (here and after, it's also denoted as $\{\phi_1,   \cdots,   \phi_s\}$ where $\phi_i = |\phi_i\rangle\langle \phi_i|$'s are the corresponding density operators) are PPT-dinstinguishable if and only if the following semidefinite program:
\begin{equation}\begin{split}
\alpha  =  & \max_{_{P_1,   \cdots,   P_s}} \frac1s \sum_{k=1}^s Tr(P_k \phi_k) \\
                        s.  t.   &\phantom{=}P_1 + \cdots + P_s = I_{_{AB}} \\
                             &\phantom{=}P_1,   \cdots,   P_s \geq 0 \\
                             &\phantom{=}\mathrm{T}_{A}(P_1),   \cdots,   \mathrm{T}_{A}(P_s) \geq 0
\end{split}\end{equation}
has optimal value $\alpha = 1$ \cite{Cosentino13}. Here we also apply this semidefinte program to construct genuinely nonlocal sets on multipartite systems. For a set of orthogonal quantum states $\mathcal{S} = \{ |\Phi_1\rangle,   \cdots,   |\Phi_s\rangle\}$ on three-partite system $\mathcal{H}_A \otimes \mathcal{H}_B \otimes \mathcal{H}_C$, a sufficient condition for their genuine nonlocality is that across each bipartition these states are PPT-indistinguishable. More explicitly, consider the family of semidefinite programs
\begin{equation}\begin{split}\label{primal}
\alpha_{X} & = \max_{_{P_1^{(X)},   \cdots,   P_s^{(X)}}} \frac1s \sum_{k=1}^s Tr(P_k^{(X)} \Phi_k ) \\
                        s.  t.   &\phantom{=}P_1^{(X)} + \cdots + P_s^{(X)} = I_{_{ABC}}, \\
                             &\phantom{=}P_1^{(X)},   \cdots,   P_s^{(X)} \geq 0, \\
                             &\phantom{=}\mathrm{T}_{X}(P_1^{(X)}),   \cdots,   \mathrm{T}_{X}(P_s^{(X)}) \geq 0
\end{split}\end{equation}
where $X = A, B$ or $C$ represents that we are dealing with the bipartition $A|BC$, $B|CA$ or $C|AB$ respectively. If for a specific set $\mathcal{S}$ we have $\alpha_{X} < 1$ for $\forall X \in \{A, B, C\}$, namely, PPT-indistinguishable through all the bipartitions, then the set is genuinely nonlocal. For convenience of calculation, we further consider the dual problems:
\begin{equation}\begin{split}\label{dual}
\beta_{X}  & =  \min_{_{Y^{(X)},   Q_1^{(X)},   \cdots,   Q_s^{(X)}}} \frac1s \ Tr(Y^{(X)}) \\
                        s.  t.   &\phantom{=}Y^{(X)} - \Phi_k \geq \mathrm{T}_X (Q_k^{(X)}),\\
                                  &\phantom{=}Q_k^{(X)} \geq 0  \ \ \ \ \ \ \ \ (1 \leq k \leq s)
\end{split}\end{equation}
where $X \in \{A, B, C\}$. By the Slater's condition,   we know that strong duality holds: $\alpha_{X} = \beta_{X} \ (X \in \{A, B, C\})$. Therefore, given a set $\mathcal{S} = \{ |\Phi_1\rangle,   \cdots,   |\Phi_s\rangle\}$ of orthogonal three-partite states, once we found that $\beta_{X} < 1$ for $\forall X \in \{A, B, C\}$,  then genuine nonlocality of $\mathcal{S}$ will be concluded.

To use these semidefinite programs to construct genuinely nonlocal sets of GHZ-lattices, notice that for subsystems $X \in \{A, B, C\}$, $\mathrm{T}_{X_1 \cdots X_t}(\Psi_{\overrightarrow{v}}) = \mathrm{T}_{X_1}(\psi_{v_1})\otimes \cdots \otimes \mathrm{T}_{X_t}(\psi_{v_t})$, where $\psi_{v_i} = |\psi_{v_i}\rangle\langle \psi_{v_i}|$ ($\Psi_{\overrightarrow{v}} = |\Psi_{\overrightarrow{v}}\rangle\langle \Psi_{\overrightarrow{v}}|$) are the corresponding density operators of the three-qubit GHZ basis (resp. $t$-level GHZ-lattices). The subscripts ``$_{X_1 \cdots X_t}$'' indicate that $t$ qubits are being held by $X = A, B$ or $C$ and we sometimes abbreviate it as ``$_X$'' when  no ambiguity occurs. We should also note that the partial-transpose operations $\mathrm{T}_X$'s are linear maps on $\mathcal{L}(\mathcal{H}_{A} \otimes \mathcal{H}_{B} \otimes \mathcal{H}_{C})$, the space of all linear operators on $\mathcal{H}_{A} \otimes \mathcal{H}_{B} \otimes \mathcal{H}_{C}$. Moreover, when GHZ-lattices  are considered, we have the following property about $\mathrm{T}_A, \mathrm{T}_B$ and $\mathrm{T}_C$ that is crucial for all subsequent discussions.

\begin{Lem}\label{lemma1}
In the linear space $\mathcal{L}(\mathbb{C}^2 \otimes \mathbb{C}^2 \otimes \mathbb{C}^2)$ of all linear operators on $\mathbb{C}^2 \otimes \mathbb{C}^2 \otimes \mathbb{C}^2$, the linear maps $\mathrm{T}_A$, $\mathrm{T}_B$ and $\mathrm{T}_C$ act invariantly on the subspace spanned by $\{\psi_0, \cdots, \psi_7\}$. Moreover, on this subspace they have matrix representations:
\begin{equation}\begin{split}\label{transform}
& \mathrm{T}_A\begin{bmatrix}
\psi_0 \\ \psi_7 \\ \psi_3 \\ \psi_4 \\ \psi_1 \\ \psi_6 \\ \psi_2 \\ \psi_5
\end{bmatrix}
= \arraycolsep = 0.01em \left[\begin{array}{cc}
\begin{matrix}
& &\frac12 & &\frac12 & &\frac12 &-&\frac12 \\
& &\frac12 & &\frac12 &-&\frac12 & &\frac12 \\
& &\frac12 &-&\frac12 & &\frac12 & &\frac12 \\
&-&\frac12 & &\frac12 & &\frac12 & &\frac12 \\
\end{matrix} & 0 \\
0 & \begin{matrix}
& &\frac12 & &\frac12 & &\frac12 &-&\frac12 \\
& &\frac12 & &\frac12 &-&\frac12 & &\frac12 \\
& &\frac12 &-&\frac12 & &\frac12 & &\frac12 \\
&-&\frac12 & &\frac12 & &\frac12 & &\frac12 \\
\end{matrix}
\end{array}\right] \begin{bmatrix}
\psi_0 \\ \psi_7 \\ \psi_3 \\ \psi_4 \\ \psi_1 \\ \psi_6 \\ \psi_2 \\ \psi_5
\end{bmatrix}, \\
\\
& \mathrm{T}_B\begin{bmatrix}
\psi_0 \\ \psi_7 \\ \psi_2 \\ \psi_5 \\ \psi_3 \\ \psi_4 \\ \psi_1 \\ \psi_6
\end{bmatrix}
= \arraycolsep = 0.01em \left[\begin{array}{cc}
\begin{matrix}
& &\frac12 & &\frac12 & &\frac12 &-&\frac12 \\
& &\frac12 & &\frac12 &-&\frac12 & &\frac12 \\
& &\frac12 &-&\frac12 & &\frac12 & &\frac12 \\
&-&\frac12 & &\frac12 & &\frac12 & &\frac12 \\
\end{matrix} & 0 \\
0 & \begin{matrix}
& &\frac12 & &\frac12 & &\frac12 &-&\frac12 \\
& &\frac12 & &\frac12 &-&\frac12 & &\frac12 \\
& &\frac12 &-&\frac12 & &\frac12 & &\frac12 \\
&-&\frac12 & &\frac12 & &\frac12 & &\frac12 \\
\end{matrix}
\end{array}\right] \begin{bmatrix}
\psi_0 \\ \psi_7 \\ \psi_2 \\ \psi_5 \\ \psi_3 \\ \psi_4 \\ \psi_1 \\ \psi_6
\end{bmatrix},\\
\\
& \mathrm{T}_C\begin{bmatrix}
\psi_0 \\ \psi_7 \\ \psi_1 \\ \psi_6 \\ \psi_2 \\ \psi_5 \\ \psi_3 \\ \psi_4
\end{bmatrix}
= \arraycolsep = 0.01em \left[\begin{array}{cc}
\begin{matrix}
& &\frac12 & &\frac12 & &\frac12 &-&\frac12 \\
& &\frac12 & &\frac12 &-&\frac12 & &\frac12 \\
& &\frac12 &-&\frac12 & &\frac12 & &\frac12 \\
&-&\frac12 & &\frac12 & &\frac12 & &\frac12 \\
\end{matrix} & 0 \\
0 & \begin{matrix}
& &\frac12 & &\frac12 & &\frac12 &-&\frac12 \\
& &\frac12 & &\frac12 &-&\frac12 & &\frac12 \\
& &\frac12 &-&\frac12 & &\frac12 & &\frac12 \\
&-&\frac12 & &\frac12 & &\frac12 & &\frac12 \\
\end{matrix}
\end{array}\right] \begin{bmatrix}
\psi_0 \\ \psi_7 \\ \psi_1 \\ \psi_6 \\ \psi_2 \\ \psi_5 \\ \psi_3 \\ \psi_4
\end{bmatrix}.
\end{split}\end{equation}
\end{Lem}

The proof is straightforward by routine calculation, which is shown in the Appendix. Now starting from the case $t=1$, consider the subset $\mathcal{S}_{5} = \{\psi_0, \psi_1, \psi_2, \psi_3, \psi_4\}$ of the three-qubit GHZ basis, namely $\Phi_k = \psi_{k-1} \ (k=1, \cdots, 5)$. By the matrix representations (\ref{transform}) in Lemma \ref{lemma1}, one can check what we list in the following is a set of feasible solutions to the family of dual problems (\ref{dual}) (namely, the constraints $Y^{(X)} - \Phi_k \geq \mathrm{T}_{X}(Q_k^{(X)})$ are satisfied for $1 \leq k \leq 5$ and $X \in \{A, B, C\}$):
\begin{equation}\begin{split}\label{s51}
& Y^{(A)} = \frac12\psi_0 +  \psi_1 + \psi_2 + \frac12 \psi_3 + \frac12 \psi_4 + \frac12 \psi_7,\\
& Q_1^{(A)} = \psi_4, Q_2^{(A)} = Q_3^{(A)} = 0, Q_4^{(A)} = \psi_7, Q_5^{(A)} = \psi_0;\\
\\
& Y^{(B)} = \psi_0 + \frac12 \psi_1 + \psi_2 + \frac12 \psi_3 + \frac12 \psi_4 + \frac12 \psi_6,\\
& Q_1^{(B)} = 0, Q_2^{(B)} = \psi_4, Q_3^{(B)} = 0, Q_4^{(B)} = \psi_6,  Q_5^{(B)} = \psi_1;\\
\\
& Y^{(C)} = \psi_0 + \psi_1  + \frac12 \psi_2 + \frac12 \psi_3  + \frac12 \psi_4 + \frac12 \psi_5,\\
& Q_1^{(C)} = Q_2^{(C)} = 0, Q_3^{(C)} = \psi_4, Q_4^{(C)} = \psi_5,  Q_5^{(C)} = \psi_2.
\end{split}\end{equation}
Here as an example, we check only the inequality constraints for bipartition $A|BC$ and the others are similar by routine calculation:
\begin{equation*}\begin{split}
\mathrm{T}_A(Q_1^{(A)}) & = -\frac12\psi_0 + \frac12\psi_7 + \frac12\psi_3 + \frac12\psi_4 \\
& < -\frac12\psi_0 + \frac12 \psi_7 + \frac12 \psi_3 + \frac12 \psi_4 + (\psi_1 + \psi_2)\\
& = Y^{(A)} - \psi_0, \\
\end{split}\end{equation*}
\begin{equation*}\begin{split}
\mathrm{T}_A(Q_2^{(A)}) & = 0 < \frac12\psi_0  + \psi_2 + \frac12 \psi_3 + \frac12 \psi_4 + \frac12 \psi_7 \\
& = Y^{(A)} - \psi_1, \\
\end{split}\end{equation*}
\begin{equation}\begin{split}\label{s52}
\mathrm{T}_A(Q_3^{(A)}) & = 0 < \frac12\psi_0  + \psi_1 + \frac12 \psi_3 + \frac12 \psi_4 + \frac12 \psi_7 \\
& = Y^{(A)} - \psi_2, \\
\end{split}\end{equation}
\begin{equation*}\begin{split}
\mathrm{T}_A(Q_4^{(A)}) & = \frac12\psi_0 + \frac12\psi_7 - \frac12\psi_3 + \frac12\psi_4 \\
& < \frac12\psi_0 + \frac12 \psi_7 - \frac12 \psi_3 + \frac12 \psi_4 + (\psi_1 + \psi_2)\\
& = Y^{(A)} - \psi_3, \\
\end{split}\end{equation*}
\begin{equation*}\begin{split}
\mathrm{T}_A(Q_5^{(A)}) & = \frac12\psi_0 + \frac12\psi_7 + \frac12\psi_3 - \frac12\psi_4 \\
& < \frac12\psi_0 + \frac12 \psi_7 + \frac12 \psi_3 - \frac12 \psi_4 + (\psi_1 + \psi_2) \\
& = Y^{(A)} - \psi_4.
\end{split}\end{equation*}
The target values for the three bipartitions are $4/5, 4/5, 4/5$ respectively. Therefore, the set $\mathcal{S}_{5}$ is genuinely nonlocal, which consists with our former discussion. Now with these, we give a simple construction of a genuinely nonlocal set of GHZ-lattices in $\mathbb{C}^4 \otimes \mathbb{C}^4 \otimes \mathbb{C}^4$, of which the cardinality is considerably smaller than $4^2+1$ (which is a trivial cardinality by Lemma \ref{lemma0}).
\bigskip
\begin{Prop}\label{s10}
In $\mathbb{C}^4 \otimes \mathbb{C}^4 \otimes \mathbb{C}^4$, the following set of orthogonal lattice-type GHZ states
\begin{equation*}\begin{split}
\mathcal{S}_{10} = \{ \psi_0 \otimes \psi_0, \psi_1 \otimes \psi_0, \psi_2 \otimes \psi_0, \psi_3 \otimes \psi_0, \psi_4 \otimes \psi_0,\\
\psi_0 \otimes \psi_7, \psi_1 \otimes \psi_7, \psi_2 \otimes \psi_7, \psi_3 \otimes \psi_7, \psi_4 \otimes \psi_7 \}
\end{split}\end{equation*}
is genuinely nonlocal.
\end{Prop}

\begin{proof}
Notice that by matrix representations (\ref{transform}), we have $\mathrm{T}_{X}(\psi_0 + \psi_7) = \psi_0 + \psi_7$ for $\forall X \in \{A, B, C\}$. From the above discussion, the constraints $Y^{(X_1)} - \psi_{k-1} \geq \mathrm{T}_{X_1}(Q_k^{(X_1)}) $ are satisfied for $1 \leq k \leq 5$ and $X_1 \in \{A_1, B_1, C_1\}$, so we have
\begin{equation*}\begin{split}
&Y^{(X_1)} \otimes (\psi_0 + \psi_7) - \psi_{k-1} \otimes (\psi_0 + \psi_7) \\
\geq & \ \mathrm{T}_{X_1}(Q_k^{(X_1)}) \otimes (\psi_0 + \psi_7) \\
= & \ \mathrm{T}_{X_1X_2}[Q_k^{(X_1)} \otimes (\psi_0 + \psi_7)].
\end{split}\end{equation*}
If we set $Y^{(X_1X_2)} =Y^{(X_1)} \otimes (\psi_0 + \psi_7)$ and  $Q_k^{(X_1X_2)} = Q_{k+5}^{(X_1X_2)} = Q_k^{(X_1)} \otimes (\psi_0 + \psi_7) \ (1 \leq k \leq 5)$, then
$$Y^{(X_1X_2)} - \Psi_j > \mathrm{T}_{X_1X_2}(Q_j^{(X_1X_2)}) \ \ (j = 1, \cdots, 10)$$
also satisfy the inequality constraints of dual problems (\ref{dual}), where
\begin{equation*}
\Psi_j = \left\{
\begin{array}{ll}
\psi_{j-1} \otimes \psi_0, & 1 \leq j \leq 5  \\
\psi_{j-6} \otimes \psi_7, & 6 \leq j \leq 10
\end{array} \right.
\end{equation*}
are just elements of $\mathcal{S}_{10}$. Since $Tr(Y^{(A)}) = Tr(Y^{(B)}) = Tr(Y^{(C)}) = 8$, all the minimums $\beta_X \ (X \in \{A, B, C\})$ are smaller than 1. This means that $\mathcal{S}_{10} = \{|\Psi_1\rangle, \cdots, |\Psi_{10}\rangle\}$ is genuinely nonlocal.
\end{proof}

The readers may be curious of the above construction procedure about $\mathcal{S}_{5}$ and $\mathcal{S}_{10}$, for the feasible solutions to semidefinte programs (\ref{dual}) we've presented (the $Y^{(X)}$'s and the $Q_k^{(X)}$'s) are all lattice-type operators (diagonal in the basis $\{|\psi_0\rangle, \cdots, |\psi_7\rangle\}$ or their tensor products). To explain why, we present the following lemma that is an analogue to Theorem 2 of \cite{Cosentino13}. It turns out when GHZ-lattices are considered, any set of feasible solutions to the semidefinite programs (\ref{dual}) corresponds to a set of lattice-type ones that have the same target values.
\bigskip
\begin{Lem}\label{lemma2}
If the states $|\Phi_1\rangle, \cdots, |\Phi_s\rangle$ to be distinguished are all GHZ-lattices, then the semidefinite programs (\ref{dual}) [and (\ref{primal}) likewise] can be reduced to linear programs.
\end{Lem}
\begin{proof}
We check here the dual problem case and the primal case is likewise. First, note that the set of three-qubit GHZ basis (\ref{GHZ}) is the common eigenbasis of the following set of commutative product operators on $\mathbb{C}^2 \otimes \mathbb{C}^2 \otimes \mathbb{C}^2$:
$$C = \{\mathrm{III}, \mathrm{XXX}, \mathrm{IZZ}, \mathrm{ZIZ}, \mathrm{ZZI}, \mathrm{XYY}, \mathrm{YXY}, \mathrm{YYX}\},$$
where $\{\mathrm{X}, \mathrm{Y}, \mathrm{Z}\}$ are the one-qubit Pauli operators. More explicitly, they have spetrum decomposition:
\begin{equation}\begin{array}{cl}\label{spectrum}
 \mathrm{III} & = \phantom{+}\psi_0 + \psi_1 + \psi_2 + \psi_3 + \psi_4 + \psi_5 + \psi_6 + \psi_7,\\
 \mathrm{XXX} & = \phantom{+}\psi_0 + \psi_1 + \psi_2 + \psi_3 - \psi_4 - \psi_5 - \psi_6 - \psi_7,\\
 \mathrm{IZZ} & = \phantom{+}\psi_0 - \psi_1 - \psi_2 + \psi_3 + \psi_4 - \psi_5 - \psi_6 + \psi_7,\\
 \mathrm{ZIZ} & = \phantom{+}\psi_0 - \psi_1 + \psi_2 - \psi_3 - \psi_4 + \psi_5 - \psi_6 + \psi_7,\\
 \mathrm{ZZI} & = \phantom{+}\psi_0 + \psi_1 - \psi_2 - \psi_3 - \psi_4 - \psi_5 + \psi_6 + \psi_7,\\
 \mathrm{XYY} & = - \psi_0 + \psi_1 + \psi_2 - \psi_3 + \psi_4 - \psi_5 - \psi_6 + \psi_7,\\
 \mathrm{YXY} & = - \psi_0 + \psi_1 - \psi_2 + \psi_3 - \psi_4 + \psi_5 - \psi_6 + \psi_7,\\
 \mathrm{YYX} & = - \psi_0 - \psi_1 + \psi_2 + \psi_3 - \psi_4 - \psi_5 + \psi_6 + \psi_7.
\end{array}\end{equation}
Consider the quantum operation
\begin{equation}\label{channel}
\Delta(\rho) = \frac{1}{|C|} \sum_{U_i \in C} U_i \rho U_{i}^{\dag}
\end{equation}
that acts on system $\mathbb{C}^2 \otimes \mathbb{C}^2 \otimes \mathbb{C}^2$. Substituting (\ref{spectrum}) into (\ref{channel}), one can check that all the cross terms in expression eliminate and hence $\Delta$ acts as the dephasing operation under the three-qubit GHZ basis:
\begin{equation} \Delta(\rho) = \sum_{i=0}^{7} |\psi_{i}\rangle\langle\psi_{i}| \rho |\psi_{i}\rangle\langle\psi_{i}|.\end{equation}
It's obvious that $\Delta$ is positive and trace-preserving. What's more, it commutes with the partial transpose operations $\mathrm{T}_{X}$'s $ (X \in \{A, B, C\})$. To see this, just write $\Delta$ in the form (\ref{channel}), and check if each $\mathrm{T}_X$ commutes with each term $U_i \rho U_i^{\dag}$. This is routine, for example
\begin{equation}\begin{split}
\mathrm{T}_B(\mathrm{YYX} \cdot \rho \cdot \mathrm{YYX}) & = \mathrm{YY^{T}X} \cdot \mathrm{T}_B(\rho) \cdot \mathrm{YY^{T}X} \\
& = (-\mathrm{YYX})\cdot \mathrm{T}_B(\rho)\cdot(-\mathrm{YYX}) \\
& = \mathrm{YYX} \cdot \mathrm{T}_B(\rho)\cdot \mathrm{YYX},
\end{split}\end{equation}
and so on. Hence, $\mathrm{T}_X[\Delta(\rho)] = \Delta[\mathrm{T}_X(\rho)]$ for $X \in \{A, B, C\}$. Now let $\Lambda = \Delta^{\otimes t}$, that is
\begin{equation*}
\Lambda(\gamma) = \sum_{i_1 = 0}^7 \cdots\sum_{i_t = 0}^7 \gamma_{i_1 \cdots i_t} \cdot |\psi_{i_1}\cdots\psi_{i_t}\rangle\langle\psi_{i_1}\cdots\psi_{i_t}|,
\end{equation*}
where $\gamma$ is any state in tripartite system $\mathbb{C}^{2^t} \otimes \mathbb{C}^{2^t} \otimes \mathbb{C}^{2^t}$ and $\gamma_{i_1 \cdots i_t} = \langle\psi_{i_1}\cdots\psi_{i_t}| \gamma |\psi_{i_1}\cdots\psi_{i_t}\rangle$. It also acts as the dephasing operation under basis $\{|\psi_{i_1}\rangle \otimes \cdots \otimes |\psi_{i_t}\rangle \}_{i_1, \cdots, i_t}$ and has the same property as $\Delta(\cdot)$. Since this operation is invariant on the GHZ-lattices $\Phi_k = \psi_{k_1} \otimes \cdots \otimes \psi_{k_t}$, by acting it on any specific feasible solution of (\ref{dual}), one can achieve another feasible solution consisting only of diagonal operators while the target values stay still. That is to say, the semidefinite programs (\ref{dual}) reduce to linear programs
\begin{equation}\begin{split}\label{linear}
\beta_X  = &  \min  \frac1s \sum_{l_1=0}^7 \cdots \sum_{l_t=0}^7 (\overrightarrow{y^{_{(X)}}})_{l_1\cdots l_t}  \\
                s.  t. \ \ & (\overrightarrow{y^{_{(X)}}})_{i_1\cdots i_t}  - \sum_{l_1\cdots l_t}(T^{_{(X)}})^{\otimes t}_{i_1\cdots i_t,  l_1\cdots l_t} \cdot (\overrightarrow{q_k^{(X)}})_{l_1\cdots l_t}  \\
                & \geq \delta_{i_1\cdots i_t, k_1\cdots k_t}, \\
                             &  (\overrightarrow{q_k^{(X)}})_{i_1\cdots i_t} \geq 0 \\
        &\phantom{(\overrightarrow{q_k^{(X)}})_{i_1\cdots i_t}} (i_1,\cdots, i_t \in \{0, \cdots, 7\}; 1 \leq k \leq s )
\end{split}\end{equation}
where $X \in\{A, B, C\}$. Therein, ``$\delta$'' is the Kronecker delta and we suppose that the GHZ-lattices $|\Phi_k\rangle$ have index $k_1, \cdots, k_t \ (1 \leq k \leq s)$; The $(\overrightarrow{y^{_{(X)}}})$'s, $(\overrightarrow{q_k^{_{(X)}}})$'s ($1 \leq k \leq s, \ X \in \{A, B, C\}$) are $8^t \times 1$ vectors representing the diagonal elements of $Y^{(X)}$'s and $Q_k^{(X)}$'s in problem (\ref{dual}) respectively (which are the optimization variables); $\mathrm{T}^{(X)}$'s $(X \in \{A, B, C\})$ are the tranformation matrices of linear maps $\mathrm{T}_X$'s under basis $
\{\psi_0, \cdots, \psi_7\}$, as shown by (\ref{transform}) in Lemma \ref{lemma1} (but with the basis rearranged in ascending order).

\end{proof}

This result greatly reduces the complexity of the problem, for all operators we need to consider now are just diagonal (lattice-type) ones. We should also point out here that the feasible solutions we presented in the case $\mathcal{S}_{5}$  are actually optimal.  Moreover, by running the linear programs (\ref{linear}) numerically, we found the existence of genuinely nonlocal subsets of $\mathcal{S}_{10}$ with cardinality down to 7. In other words, the genuine nonlocality of $\mathcal{S}_{10}$ can actually be derived from its subset. 

\begin{Thm}\label{s7}
In $\mathbb{C}^4 \otimes \mathbb{C}^4 \otimes \mathbb{C}^4$, the following $7$-ary subset of $\mathcal{S}_{10}$:
\begin{equation*}\begin{split}
\mathcal{S}_{7} = \{ & \psi_0 \otimes \psi_0, \psi_1 \otimes \psi_0, \psi_2 \otimes \psi_0, \\
& \psi_3 \otimes \psi_0, \psi_4 \otimes \psi_0, \psi_3 \otimes \psi_7, \psi_4 \otimes \psi_7 \} \\
\end{split}\end{equation*}
already has the property of genuine nonlocality ($\psi_0 \otimes \psi_7, \psi_1 \otimes \psi_7, \psi_2 \otimes \psi_7$ from $\mathcal{S}_{10}$ have been excluded).
\end{Thm}

The proof is placed in the Appendix, where we present a set of feasible solutions to the dual problems (\ref{dual}) with target values smaller than $1$. Notably, although these solutions come out of numerical evaluation, it turns out that they can also be constructed inductively from the feasible solutions (\ref{s51}) for $\mathcal{S}_{5}$. This inspire us to further construct such genuinely nonlocal sets inductively when $t>2$, simply duplicating the procedure from $\mathcal{S}_{5}$ to $\mathcal{S}_{7}$. That is to say, in $\mathbb{C}^{8} \otimes \mathbb{C}^{8} \otimes \mathbb{C}^{8}$ we will achieve
\begin{equation*}\begin{split}
\mathcal{S}_{11} = \{ & \psi_0 \otimes \psi_0 \otimes \psi_0, \psi_1 \otimes \psi_0 \otimes \psi_0, \psi_2 \otimes \psi_0 \otimes \psi_0, \\
& \psi_3 \otimes \psi_0 \otimes \psi_0, \psi_4 \otimes \psi_0 \otimes \psi_0, \psi_3 \otimes \psi_7 \otimes \psi_0, \\
& \psi_4 \otimes \psi_7 \otimes \psi_0, \psi_3 \otimes \psi_0 \otimes \psi_7, \psi_4 \otimes \psi_0 \otimes \psi_7, \\
& \psi_3 \otimes \psi_7 \otimes \psi_7, \psi_4 \otimes \psi_7 \otimes \psi_7 \}
\end{split}\end{equation*}
(see (\ref{st}) in the proof of Theorem \ref{d3}) and more generally:
\begin{Thm}\label{d3}
In three-partite system $\mathbb{C}^{d} \otimes \mathbb{C}^{d} \otimes \mathbb{C}^{d}$ where $d = 2^t (t \geq 1)$, there exist genuinely nonlocal sets of lattice-type GHZ states with cardinality $d+3$.
\end{Thm}

\begin{proof}
Inspired by the proof of Theorem \ref{s7}, we explicitly construct these sets $\mathcal{S}^{(t)}$ by mathematical induction on $t$. For $t=1, 2$, they are just $\mathcal{S}_5$, $\mathcal{S}_7$ in the main text. Now assume that $\mathcal{S}^{(t)} = \{\eta_1, \cdots, \eta_{2^t+3}\}$, together with the $Y^{(X_1\cdots X_t)}$, $Q_j^{(X_1\cdots X_t)} \ (1 \leq j \leq d+3$ where $d = 2^t)$ have been constructed, such that
\begin{equation}\begin{split}\label{20}
Y^{(A_1\cdots A_t)} & - \eta_2 - \eta_3 \\
\geq \ 0 = & \mathrm{T}_{A}(Q_{2}^{(A_1 \cdots A_t)}) = \mathrm{T}_{A}(Q_{3}^{(A_1 \cdots A_t)}), \\
Y^{(B_1\cdots B_t)} & - \eta_1 - \eta_3 \\
\geq \ 0 = & \mathrm{T}_{B}(Q_{1}^{(B_1 \cdots B_t)}) = \mathrm{T}_{B}(Q_{3}^{(B_1 \cdots B_t)}), \\
Y^{(C_1\cdots C_t)} & - \eta_1 - \eta_2 \\
\geq \ 0 = & \mathrm{T}_{C}(Q_{1}^{(C_1 \cdots C_t)}) = \mathrm{T}_{C}(Q_{2}^{(C_1 \cdots C_t)}), \\
\end{split}\end{equation}
and
\begin{equation}\begin{split}\label{21}
Y^{(A_1\cdots A_t)} - \eta_2 - \eta_3 \geq \mathrm{T}_{A}(Q_{j}^{(A_1 \cdots A_t)}) +& \eta_j \\
                                                                                                   (1 \leq j \leq d+3, \ & j \neq 2, 3),\\
Y^{(B_1\cdots B_t)} - \eta_1 - \eta_3 \geq \mathrm{T}_{B}(Q_{j}^{(B_1 \cdots B_t)}) +& \eta_j \\
                                                                                                   (1 \leq j \leq d+3, \ & j \neq 1, 3),\\
Y^{(C_1\cdots C_t)} - \eta_1 - \eta_2 \geq \mathrm{T}_{C}(Q_{j}^{(C_1 \cdots C_t)}) +& \eta_j \\
                                                                                                   (1 \leq j \leq d+3, \ & j \neq 1, 2).\\
\end{split}\end{equation}
Moreover, assume that $Tr[Y^{(A_1\cdots A_t)}] = Tr[Y^{(B_1\cdots B_t)}] = Tr[Y^{(C_1\cdots C_t)}] = 2^t + 2 = d+2$. It can be seen from (\ref{s51}), (\ref{s52}) and also the proof of Theorem \ref{s7} that these assumptions hold for $t = 1$.

Now, we let  $\mathcal{S}^{(t+1)} = \{\chi_1, \cdots, \chi_{2^{t+1}+3}\}$, where
\begin{equation}\label{st}
\chi_j = \left\{
\begin{array}{ll}
\eta_j \otimes \psi_0, & 1 \leq j \leq d+3  \\
\eta_{j-d} \otimes \psi_7, & d+4 \leq j \leq 2d+3.
\end{array} \right.
\end{equation}
Meanwhile, we set
\begin{equation*}\begin{split}
Y^{(A_1\cdots A_{t+1})} = Y^{(A_1\cdots A_t)} \otimes (\psi_0 + \psi_7) - \eta_2 \otimes \psi_7 - \eta_3 \otimes \psi_7, \\
Y^{(B_1\cdots B_{t+1})} = Y^{(B_1\cdots B_t)} \otimes (\psi_0 + \psi_7) - \eta_1 \otimes \psi_7 - \eta_3 \otimes \psi_7, \\
Y^{(C_1\cdots C_{t+1})} = Y^{(C_1\cdots C_t)} \otimes (\psi_0 + \psi_7) - \eta_1 \otimes \psi_7 - \eta_2 \otimes \psi_7, \\
\end{split}\end{equation*}
and
\begin{equation*}\begin{split}
& Q_j^{(X_1\cdots X_{t+1})} \\
= & \left\{
\begin{array}{ll}
Q_j^{(X_1\cdots X_{t})} \otimes (\psi_0 + \psi_7), & 1 \leq j \leq d+3  \\
Q_{j-d}^{(X_1\cdots X_{t})} \otimes (\psi_0 + \psi_7), & d+4 \leq j \leq 2d+3
\end{array} \right.
\end{split}\end{equation*}
where $X \in \{A, B, C\}$. Acting $\otimes (\psi_0 + \psi_7)$ on the r.h.s. of (\ref{20}) and (\ref{21}), we get
\begin{equation*}\begin{split}
Y^{(A_1\cdots A_{t+1})} & - \chi_2 - \chi_3 \\
\geq \ 0 = & \mathrm{T}_{A}(Q_{2}^{(A_1 \cdots A_{t+1})}) = \mathrm{T}_{A}(Q_{3}^{(A_1 \cdots A_{t+1})}), \\
Y^{(B_1\cdots B_{t+1})} & - \chi_1 - \chi_3 \\
\geq \ 0 = & \mathrm{T}_{B}(Q_{1}^{(B_1 \cdots B_{t+1})}) = \mathrm{T}_{B}(Q_{3}^{(B_1 \cdots B_{t+1})}), \\
Y^{(C_1\cdots C_{t+1})} & - \chi_1 - \chi_2 \\
\geq \ 0 = & \mathrm{T}_{C}(Q_{1}^{(C_1 \cdots C_{t+1})}) = \mathrm{T}_{C}(Q_{2}^{(C_1 \cdots C_{t+1})}), \\
\end{split}\end{equation*}
and
\begin{equation*}\begin{split}
& Y^{(A_1\cdots A_{t+1})} - \chi_2 - \chi_3 \\
\geq \ & \mathrm{T}_{A}(Q_{j}^{(A_1 \cdots A_{t+1})}) + \eta_j \otimes \psi_0 \\
& \left(\text{also} \ \ \mathrm{T}_{A}(Q_{j}^{(A_1 \cdots A_{t+1})}) + \eta_j \otimes \psi_7\right) \\
& \phantom{\mathrm{T}_{A}(Q_{j}) \otimes \otimes \otimes} (1 \leq j \leq d+3, \ j \neq 2, 3),\\
\end{split}\end{equation*}
\begin{equation*}\begin{split}
& Y^{(B_1\cdots B_{t+1})} - \chi_1 - \chi_3 \\
\geq \ & \mathrm{T}_{B}(Q_{j}^{(B_1 \cdots B_{t+1})}) + \eta_j \otimes \psi_0 \\
& \left(\text{also} \ \ \mathrm{T}_{B}(Q_{j}^{(B_1 \cdots B_{t+1})}) + \eta_j \otimes \psi_7\right) \\
& \phantom{\mathrm{T}_{A}(Q_{j}) \otimes \otimes \otimes} (1 \leq j \leq d+3, \ j \neq 1, 3),\\
\end{split}\end{equation*}
\begin{equation*}\begin{split}
& Y^{(C_1\cdots C_{t+1})} - \chi_1 - \chi_2 \\
\geq \ & \mathrm{T}_{C}(Q_{j}^{(C_1 \cdots C_{t+1})}) + \eta_j \otimes \psi_0 \\
& \left(\text{also} \ \ \mathrm{T}_{C}(Q_{j}^{(C_1 \cdots C_{t+1})}) + \eta_j \otimes \psi_7\right)  \\
& \phantom{\mathrm{T}_{A}(Q_{j}) \otimes \otimes \otimes} (1 \leq j \leq d+3, \  j \neq 1, 2).\\
\end{split}\end{equation*}
That is, the assumptions (\ref{20}) and (\ref{21}) also hold for $\mathcal{S}^{(t+1)}$. Besides, $Tr[Y^{(A_1\cdots A_{t+1})}] = Tr[Y^{(B_1\cdots B_{t+1})}] = Tr[Y^{(C_1\cdots C_{t+1})}] = 2(d+2)-2 = 2^{t+1}+2 < 2^{t+1}+3$. Since the constraints of the dual problems (\ref{dual}) are all satisfied for any $\mathcal{S}^{(t)} \ (t \geq 1)$, these sets are consequently genuinely nonlocal.
\end{proof}

This result is somewhat unexpected, for the size of such genuinely nonlocal sets can scale down to linear in the local dimension $d$, conspicuously with a small linear factor $l = 1$!  In sharp contrast, when strong nonlocality are considered, all existing nonlocal sets that've been constructed in three-partite system $\mathbb{C}^d \otimes \mathbb{C}^d \otimes \mathbb{C}^d$  have cardinality $\Theta(d^2)$ (with the leading coefficient $c>1$), in our best knowledge \cite{Yuan20,Wang21,Shi2022,Shi22,Li22arxiv}. This comparison might arguably imply a substantial difference between strong nonlocality and distinguishability-based genuine nonlocality.

\section{Constructing $(2,3)$-threshold sets in three-partite systems}\label{section4}

In the previous section, we have studied the genuine nonlocality of the GHZ states on three-qudit systems. Namely, the distinguishability when only two of the three parties $A, B, C$ are allowed to combine. In the $n$-partite scenario, if the system is in one of a priorly known set of quantum states that are genuinely nonlocal, then the actual state of the system cannot be revealed perfectly unless all $n$ parties are joint together to make global measurements. It's then natural and interesting to extend this notion of multipartite distinguishability to the more general case: a set of orthogonal states in the $n$-partite system is called \textit{$(s,n)$-threshold distinguishable} if the states are distinguishable when any $s$ or more parties are combined ($s \leq n$), but indistinguishable when at most $s-1$ parties can combine. Such sets are called \textit{$(s,n)$-threshold sets} here and genuinely nonlocal sets are just $(n,n)$-threshold sets in this sense. This notion of multipartite distinguishability is an analogue to the resembling notion in quantum secret sharing (QSS), where a \textit{$(s,n)$-threshold QSS scheme} means that the secret have been distributed into $n$ shares and any group of $s$ or more shares can collaboratively reconstruct the secret while no group of fewer than $s$ shares can \cite{Cleve1999,Gottesman2000,Rahaman2015,Wang2017}.

In this section, we consider the problem of constructing $(2,3)$-threshold sets consisting of GHZ states in three-partite systems, as a complement to our studies of $(3,3)$-threshold sets in Sec. \ref{section3}. It's routine to first consider the simplest case of three-qubit system. Unfortunately, it turns out that $(2,3)$-threshold sets consisting of states from the three-qubit GHZ basis do not exist:

\begin{Prop}\label{nt32}
In system $\mathbb{C}^{2} \otimes \mathbb{C}^{2} \otimes \mathbb{C}^{2}$, there exists no $(2,3)$-threshold set consisting of states from the three-qubit GHZ basis (\ref{GHZ}).
\end{Prop}

\begin{proof}
Suppose that $\mathcal{T}$ is a $(2,3)$-threshold set consisting of states from the three-qubit GHZ basis, namely, it is distinguishable across all bipartitions $A|BC$, $B|CA$ and $C|AB$ but indistinguishable when the three parties are fully separated. Now if $\mathcal{T}$ contains a certain conjuate pair, one can easily see from the schematic picture in FIG \ref{fig} that the set cannot contain a third state, for otherwise $\mathcal{T}$ will be indistinguishable across one of the bipartitions. However, all conjugate pairs are distinguishable by the three fully separated paries. As an example, for $|\psi_{0,7}\rangle=\frac{1}{\sqrt{2}} (|000\rangle \pm |111\rangle)$, since
$$\frac{|000\rangle + |111\rangle}{\sqrt{2}} = \frac{|+++\rangle + |+--\rangle + |-+-\rangle + |--+\rangle}{2}$$
and
$$\frac{|000\rangle - |111\rangle}{\sqrt{2}} = \frac{|++-\rangle + |+-+\rangle + |-++\rangle + |---\rangle}{2}$$
$$\left( |+\rangle = \frac{|0\rangle+|1\rangle}{\sqrt{2}}, \ \ |-\rangle = \frac{|0\rangle-|1\rangle}{\sqrt{2}}\right),$$
the three parties can do local POVM $\{|+\rangle\langle +|, |-\rangle\langle -|\}$ on their own subsystems and tell them apart by classical communications. Therefore, $\mathcal{T}$ must not contain any conjugate pair. Unfortunately, when no conjugate pair exists in $\mathcal{T}$, the three fully separated parties can distinguish the set by simply doing local POVM $\{|0\rangle\langle0|, |1\rangle\langle1|\}$ on their subsystems and communicating classically. As a result, no $(2,3)$-threshold set exists in this case.
\end{proof}

Although such $(2,3)$-threshold sets consisting of GHZ states do not exist in $\mathbb{C}^{2} \otimes \mathbb{C}^{2} \otimes \mathbb{C}^{2}$, it's natural to further consider their existence in systems of higher local dimension. Fortunately, by applying the method of semidefinite program similarly as before, we find the existence of $(2,3)$-threshold sets consisting of lattice-type GHZ states. If a set of orthogonal quantum states $S = \{|\Phi_1\rangle, \cdots, |\Phi_s\rangle\}$ in three-partite system $\mathcal{H}_A \otimes \mathcal{H}_B \otimes \mathcal{H}_C$ is locally distinguishable across the $A|B|C$ (fully separated) partition, then the semidefinite program
\begin{equation}\begin{split}\label{SDPsplit}
& \alpha_{A|B|C}  = \max_{_{P_1,   \cdots,   P_s}} \frac1s \sum_{k=1}^s Tr(P_k \Phi_k) \\
                        s.  t.   &\phantom{=}P_1 + \cdots + P_s = I_{_{ABC}}, \\
                             &\phantom{=}P_1,   \cdots,   P_s \geq 0, \\
                             &\phantom{=}\mathrm{T}_{A}(P_k), \mathrm{T}_{B}(P_k), \mathrm{T}_{C}(P_k) \geq 0 \ \ \ (1 \leq k \leq s)
\end{split}\end{equation}
must have optimal value $\alpha_{A|B|C} = 1$. To obtain an upper bound of $\alpha_{A|B|C}$, it's natural to consider its dual problem which we present in the following lemma:

\begin{Lem}\label{Dualsplit}
The dual problem of semidefinite program (\ref{SDPsplit}) has the form:
\begin{equation}\begin{split}\label{dualsplit}
& \beta_{A|B|C} = \min_{_{Y,   Q_k^{(A)}, Q_k^{(B)},  Q_k^{(C)}}} \frac1s \ Tr(Y) \\
                        s.  t.   &\phantom{=}Y - \Phi_k \geq \mathrm{T}_A (Q_k^{(A)}) + \mathrm{T}_B (Q_k^{(B)}) + \mathrm{T}_C (Q_k^{(C)}),\\
                                  &\phantom{=}Q_k^{(A)}, Q_k^{(B)}, Q_k^{(C)} \geq 0  \ \ \ \ \ \ (1 \leq k \leq s).
\end{split}\end{equation}
\end{Lem}

For consideration of readability, we place the proof in the Appendix. Since $\alpha_{A|B|C} \leq \beta_{A|B|C}$ (actually $\alpha_{A|B|C} = \beta_{A|B|C}$ by Slater's condition), once we find a feasible solution to the dual problem (\ref{dualsplit}) with target value smaller than $1$, the indistinguishability of $S = \{|\Phi_1\rangle, \cdots, |\Phi_s\rangle\}$ across the $A|B|C$ partition will be deduced immediately. For the distinguishability across the $2-1$ bipartitions, a necessary condition is that the family of semidefinite programs (\ref{dual}) have optimal values $\beta_X = 1$ for $X \in \{A, B, C\}$ ($\alpha_X = \beta_X$ by strong duality). These facts indicate that we can search for sets $\mathcal{T}$ such that $\beta_{A|B|C} < 1$ and $\beta_X =1 \ (X \in \{A, B, C\})$ as candidates for $(2,3)$-threshold sets in three-partite systems (we should further check their distinguishability across the $2-1$ partitions as PPT-distinguishability is just a necessary condition). In our case where GHZ-lattices are considered, these semidefinite programs reduce to linear programs by Lemma \ref{lemma1}, Lemma \ref{lemma2} and the complexity will be greatly reduced. By running numerical search, we find plenty of such sets in $\mathbb{C}^{4} \otimes \mathbb{C}^{4} \otimes \mathbb{C}^{4}$ and we present here only one of them with the smallest cardinality.
\medskip
\begin{Thm}\label{t32}
In system $\mathbb{C}^{4} \otimes \mathbb{C}^{4} \otimes \mathbb{C}^{4}$, the $5$-ary set
$$\mathcal{T}_5^{(2,3)} = \{\psi_0 \otimes \psi_1, \psi_0 \otimes \psi_6, \psi_0 \otimes \psi_3, \psi_3 \otimes \psi_0, \psi_4 \otimes \psi_0\}$$
consisting of three-partite lattice-type GHZ states is a $(2,3)$-threshold set.
\end{Thm}

\begin{proof}
First, we proof that $\mathcal{T}_5^{(2,3)}$ is indistinguishable in the $1-1-1$ partition. This can be done by presenting a feasible solution to the dual problem (\ref{dualsplit}), with target value smaller than $1$ ($19/20$ in this case):
\begin{equation*}\begin{split}
Y = & \frac12 \psi_0 \otimes (\psi_1 + \psi_6) + \psi_0 \otimes \psi_3 + \frac12 (\psi_3 + \psi_4) \otimes \psi_0  \\
+ &  \frac14 \psi_1 \otimes (\psi_0 + \psi_7) + \frac14 \psi_6 \otimes (\psi_0 + \psi_7) \\
+ &  \frac14 \psi_7 \otimes (\psi_3 + \psi_4) + \frac14 \psi_0 \otimes \psi_4; \\
\end{split}\end{equation*}

\begin{equation*}\begin{split}
& Q_1^{(A)} =  Q_2^{(A)} =  Q_3^{(A)} =0, \\
& Q_4^{(A)} =  \frac12 \psi_0 \otimes \psi_4 + \frac12 \psi_7 \otimes \psi_3,\\
& Q_5^{(A)} =  \frac12 \psi_0 \otimes \psi_3 + \frac12 \psi_7 \otimes \psi_4; \\
\end{split}\end{equation*}

\begin{equation*}\begin{split}
& Q_1^{(B)} = \frac12 (\psi_0 + \psi_7) \otimes \psi_4,\\
& Q_2^{(B)} = \frac12 (\psi_0 + \psi_7) \otimes \psi_3, \ Q_3^{(B)} = 0,\\
& Q_4^{(B)} =  \frac12 \psi_6 \otimes (\psi_0 + \psi_7),\\
& Q_5^{(B)} =  \frac12 \psi_1 \otimes (\psi_0 + \psi_7); \\
\end{split}\end{equation*}

\begin{equation*}\begin{split}
& Q_1^{(C)} =  \frac12 \psi_1 \otimes \psi_7 + \frac12 \psi_6 \otimes \psi_0,\\
& Q_2^{(C)} =  \frac12 \psi_1 \otimes \psi_0 + \frac12 \psi_6 \otimes \psi_7, \\
& Q_3^{(C)} =  Q_4^{(C)} =  Q_5^{(C)} =0. \\
\end{split}\end{equation*}

\medskip
This solution is actually an optimal one but we only need to check its feasibility. As an instance, we check here the constraint $Y - \psi_0 \otimes \psi_1 \geq \mathrm{T}_A (Q_1^{(A)}) + \mathrm{T}_B (Q_1^{(B)}) + \mathrm{T}_C (Q_1^{(C)})$ (with the help of expressions (\ref{transform})):
\begin{equation*}\begin{split}
  & \mathrm{T}_A (Q_1^{(A)}) + \mathrm{T}_B (Q_1^{(B)}) + \mathrm{T}_C (Q_1^{(C)}) \\
= \ & 0 +  \frac12 (\psi_0 + \psi_7) \otimes \frac12 (\psi_3 + \psi_4 - \psi_1 + \psi_6) \\
  &+ \frac18 (\psi_0 - \psi_7 + \psi_1 + \psi_6) \otimes (\psi_0 + \psi_7 - \psi_1 + \psi_6) \\
  &+ \frac18 (-\psi_0 + \psi_7 + \psi_1 + \psi_6) \otimes (\psi_0 + \psi_7 + \psi_1 - \psi_6) \\
= \ & \frac14 (\psi_1 + \psi_6) \otimes (\psi_0 + \psi_7) + \frac14 (\psi_0 + \psi_7) \otimes (\psi_3 + \psi_4) \\
  &+ \frac12 \psi_0 \otimes \psi_6  - \frac12 \psi_0 \otimes \psi_1 \\
\leq \ & Y - \psi_0 \otimes \psi_1,
\end{split}\end{equation*}
and the others can be checked similarly. For the distinguishability of $\mathcal{T}_5^{(2,3)}$ across the $2-1$ bipartitions, we present the explicit distinguishing protocol respectively:

\medskip
(i) For $A|BC$: $BC$ can first perform 2-outcome joint measurement $\{|00\rangle\langle00|+|11\rangle\langle11|, |01\rangle\langle01|+|10\rangle\langle10|\}_{B_2C_2}$ on the second share of qubit GHZ state to reduce $\mathcal{T}_5^{(2,3)}$ into $\{\psi_0 \otimes \psi_1, \psi_0 \otimes \psi_6\}$ and $\{\psi_0 \otimes \psi_3, \psi_3 \otimes \psi_0, \psi_4 \otimes \psi_0\}$ (see FIG \ref{fig}), of which the former can apparently be distinguished across $A|BC$ and the latter can be further distinguished into $\{\psi_0 \otimes$\sout{$\psi_3$}$\}$ or $\{\psi_3 \otimes$\sout{$\psi_0$}$, \psi_4 \otimes$\sout{$\psi_0$}$\}$ if $A|BC$ measure on the second share and communicate classically (we use the delete line ``\sout{$ \ \ $}'' to denote that the share of state has collapsed after measurements). Lastly, $\{\psi_3 \otimes$\sout{$\psi_0$}$, \psi_4 \otimes$\sout{$\psi_0$}$\}$ can also be distinguished by measuring on the first share.

\medskip
(ii) For $B|CA$: $CA$ can perform joint measurement $\{|00\rangle\langle00|+|11\rangle\langle11|, |01\rangle\langle01|+|10\rangle\langle10|\}_{C_2A_2}$ on the second share of state to reduce $\mathcal{T}_5^{(2,3)}$ into $\{\psi_3 \otimes \psi_0, \psi_4 \otimes \psi_0\}$ and $\{\psi_0 \otimes \psi_1, \psi_0 \otimes \psi_6, \psi_0 \otimes \psi_3\}$, where the former set is apparently distinguishable across $B|CA$; For the latter, $B$ can use the first share $\psi_0$ as resource, to teleport his second share of qubit to $CA$ and then $CA$ can distinguish $\{$\sout{$\psi_0$}$\otimes \psi_1,$\sout{$\psi_0$}$\otimes \psi_6,$\sout{$\psi_0$}$\otimes \psi_3\}$ by their own.

\medskip
(iii) For $C|AB$: $AB$ first perform joint measurement $\{|00\rangle\langle00|+|11\rangle\langle11|, |01\rangle\langle01|+|10\rangle\langle10|\}_{A_1B_1}$ on the first share of to reduce the set into $\{\psi_3 \otimes \psi_0, \psi_4 \otimes \psi_0\}$ and $\{\psi_0 \otimes \psi_1, \psi_0 \otimes \psi_6, \psi_0 \otimes \psi_3\}$, where the former is distinguishable across $C|AB$; For the other, $AB$ can further perform measurement $\{|00\rangle\langle00|+|11\rangle\langle11|, |01\rangle\langle01|+|10\rangle\langle10|\}_{A_2B_2}$ on the second share to reduce it into $\{\psi_0 \otimes \psi_3\}$ and $\{\psi_0 \otimes \psi_1, \psi_0 \otimes \psi_6\}$, where the latter one can be distinguished simply through the second share.
\end{proof}

\medskip
It's also straightforward to construct $(2,3)$-threshold sets in systems $\mathbb{C}^{2^t} \otimes \mathbb{C}^{2^t} \otimes \mathbb{C}^{2^t}$ ($t > 2$) inductively from $\mathcal{T}_5^{(2,3)}$, similarly as that of Proposition \ref{s10}. What's more, we have also found a $16$-ary superset of $\mathcal{T}_5^{(2,3)}$ that is PPT-distinguishable in all three $2-1$ bipartitions:
\begin{equation*}\begin{split}
\mathcal{T}_{16} =\{&\psi_0 \otimes \psi_1, \psi_0 \otimes \psi_6, \psi_0 \otimes \psi_3, \psi_3 \otimes \psi_0,\\
  & \psi_0 \otimes \psi_5, \psi_5 \otimes \psi_5, \psi_5 \otimes \psi_0,\psi_4 \otimes \psi_0,\\
  & \psi_6 \otimes \psi_5, \psi_6 \otimes \psi_3, \psi_6 \otimes \psi_6,\psi_6 \otimes \psi_0,\\
  & \psi_2 \otimes \psi_2, \psi_2 \otimes \psi_3,\psi_3 \otimes \psi_3, \psi_3 \otimes \psi_1\}.\\
\end{split}\end{equation*}
Note that $16$ achieves the upper bound of cardinality of such sets indeed (it's not hard to proof that any $k > d^2$ GHZ states in $\mathbb{C}^{d} \otimes \mathbb{C}^{d} \otimes \mathbb{C}^{d}$ are PPT-indistinguishable across the $2-1$ partitions). Yet, we don't know whether $\mathcal{T}_{16}$ is locally  distinguishable in the $2-1$ bipartitions. It's then interesting to find out the maximal intermediate $(2,3)$-threshold set(s) $\mathcal{T}_5^{(2,3)} \subseteq \mathcal{T}_{max}^{(2,3)} \subseteq \mathcal{T}_{16}$ with the largest cardinality $|\mathcal{T}_{max}^{(2,3)}|$, such that the intermediate sets between $\mathcal{T}_5^{(2,3)}$ and $\mathcal{T}_{max}^{(2,3)}$ are all $(2,3)$-threshold sets immediately!

\section{Conclusion and discussion}\label{section5}
In this paper,  we have studied the distinguishability-based genuine nonlocality of a typical type of genuine multipartite entangled states --- the GHZ states. We first study the genuine nonlocality of the three-qubit GHZ basis. Then using the result in the three-qubit case, we construct genuinely nonlocal sets of ``GHZ-lattices'' in three-partite systems where the local dimension are powers of $2$. It turns out that the size of these genuinely nonlocal sets with genuine multipartite entanglement can at least scale down to linear in the local dimension $d$, with a conspicuously small linear factor $l = 1$. The concept of genuine nonlocality, which concerns the distinguishability of multipartite quantum states through any possible bipartition of the subsystems, is in our opinion more naturally arising than the recently more popular ``strong nonlocality''. However, little is known  about this form of distinguishability-based nonlocality. A natural and interesting question is to ask: in what extent strong nonlocality is stronger than the more normal genuine nonlocality? As far as we know presently, no strongly nonlocal set with such small scale has been constructed yet, no matter product states or genuinely entangled states were considered. It's therefore reasonable to argue that there might exist a significant gap between the strength of strong nonlocality and the distinguishability-based genuine nonlocality. Besides that, our result might possibly also illuminate the relation between entanglement and distinguishability in multipartite scenario, substantiating the perspective that entanglement can somehow raise difficulty in state discrimination.

We also discuss the concept of $(s,n)$-threshold distinguishability, which extend the notion of genuine nonlocality to characterize the local accessibility of global information more comprehensively. Although there's no $(2,3)$-threshold sets in $\mathbb{C}^{2} \otimes \mathbb{C}^{2} \otimes \mathbb{C}^{2}$ consisting of states from the three-qubit GHZ basis, we have found their existence in systems of higher local dimension. We stress here that the techniques we use can be easily generalized to more-partite cases. For example, in the four-partite case, when distinguishability in bipartitions is considered, we can consider semidefinite programs (\ref{primal}) and (\ref{dual}) similarly by letting $X \in \{A, B, C, D, AB, BC, AC\}$; When considering distinguishability through the $1-1-1-1$ partition, semidefinite program (\ref{SDPsplit}) can be adjusted to
\begin{equation*}\begin{split}\label{SDPsplit4}
& \alpha_{A|B|C|D}  =  \max_{_{P_1,   \cdots,   P_s}} \frac1s \sum_{k=1}^s Tr(P_k \Phi_k) \\
                        s.  t.   &\phantom{=}P_1 + \cdots + P_s = I_{_{ABCD}}, \\
                             &\phantom{=}P_1,   \cdots,   P_s \geq 0, \\
                             &\phantom{=}\mathrm{T}_{A}(P_k), \mathrm{T}_{B}(P_k), \mathrm{T}_{C}(P_k), \mathrm{T}_{D}(P_k) \geq 0, \\
                             &\phantom{=}\mathrm{T}_{AB}(P_k), \mathrm{T}_{BC}(P_k), \mathrm{T}_{AC}(P_k) \geq 0 \ \ \ \ \ (1 \leq k \leq s)
\end{split}\end{equation*}
and so hence forth. It's interesting to derive some more universal results and this will be the main focus in our future work.

There are also some questions left to be considered. For example, whether the technique we use can be generalized to more general cases where the local dimension $d$ are not powers of $2$? Can we find genuinely nonlocal sets with even smaller cardinality? What is the situation when a large number of parties are considered? Besides, notice that we've constructed the genuinely nonlocal sets or threshold sets mainly through PPT-indistinguishability in this work. It's then natural and interesting to ask whether there exist other ways to construct these sets, such that they have even smaller cardinality or some other novel properties?

\bigskip

\bigskip
\centerline{\textbf{Acknowledgments}}
\medskip
This work was supported by the National Natural Science Foundation of China under Grant No. 62272492, the Guangdong Basic and Applied Basic Research Foundation under Grant No. 2020B1515020050, Key Research and Development Project of Guangdong Province under Grant No. 2020B0303300001 and the Guangdong Basic and Applied Basic Research Foundation under Grant No. 2020B1515310016. 

\section{Appendix}

{\noindent \bf Proof of Lemma \ref{lemma1}: }
\begin{proof}
It is routine to check this result. For example, notice that $\mathrm{T}_A(\psi_0) = (|000\rangle\langle000| + |111\rangle\langle111| + |100\rangle\langle011| + |011\rangle\langle100|)/2$, and $\psi_4 = (|011\rangle\langle011| + |100\rangle\langle100| - |011\rangle\langle100| - |100\rangle\langle011|)/2$. Then by adding these two operators we eliminate the cross terms:
\begin{equation}\begin{split}\label{TA04}
& \mathrm{T}_A(\psi_0) + \psi_4 \\
& = \frac{|000\rangle\langle000| + |111\rangle\langle111| + |011\rangle\langle011| + |100\rangle\langle100|}{2}.
\end{split}\end{equation}
Moreover, $\psi_0 + \psi_7 = |000\rangle\langle000| + |111\rangle\langle111|$ while $\psi_3 + \psi_4 = |011\rangle\langle011| + |100\rangle\langle100|$.
Therefore, we get $\mathrm{T}_A(\psi_0) = (\psi_0 + \psi_7 + \psi_3 - \psi_4)/2$ as it is shown in (\ref{transform}). By acting $\mathrm{T}_A$ on (\ref{TA04}), we get $\psi_0 + \mathrm{T}_A(\psi_4) = (\psi_0 + \psi_7 + \psi_3 + \psi_4)/2$ and so $\mathrm{T}_A(\psi_4) = (-\psi_0 + \psi_7 + \psi_3 + \psi_4)/2$. Similarly,
\begin{equation*}\label{TA73}
\mathrm{T}_A(\psi_3) + \psi_7 = \mathrm{T}_A(\psi_7) + \psi_3 = \frac{\psi_0 + \psi_7 + \psi_3 + \psi_4}{2},
\end{equation*}
hence we have $\mathrm{T}_A(\psi_7) = (\psi_0 + \psi_7 - \psi_3 + \psi_4)/2$ and $\mathrm{T}_A(\psi_3) = (\psi_0 - \psi_7 + \psi_3 + \psi_4)/2$. Actually, within the bipartition $A|BC$, $\mathcal{A}_1= \{\psi_0, \psi_7, \psi_3, \psi_4\}$ is locally equivalent to the Bell basis (see FIG \ref{fig}). Similarly for $\mathcal{A}_2 = \{\psi_1, \psi_6, \psi_2, \psi_5\}$, one can easily check that
\begin{equation*}\begin{split}
& \mathrm{T}_A(\psi_1) + \psi_5 = \mathrm{T}_A(\psi_5) + \psi_1 \\
= \ & \mathrm{T}_A(\psi_2) + \psi_6 = \mathrm{T}_A(\psi_6) + \psi_2 \\
= \ & \frac{|001\rangle\langle001| + |110\rangle\langle110| + |010\rangle\langle010| + |101\rangle\langle101|}{2} \\
= \ & \frac{\psi_1 + \psi_6 + \psi_2 + \psi_5}{2}.
\end{split}\end{equation*}
Hence, $\mathrm{T}_A$ has same matrix representation on $\mathcal{A}_2$ and so forth for the other two bipartions.
\end{proof}

\bigskip
{\noindent \bf Proof of Theorem \ref{s7}: }
\begin{proof}
We present here a set of feasible solutions (actually optimal ones) to the dual problems for $\mathcal{S}_{7}$, inductively from the solutions (\ref{s51}) for $\mathcal{S}_{5}$:

\noindent For $A|BC$,
\begin{equation*}\begin{split}
Y^{(A_1A_2)} = & Y^{(A_1)} \otimes (\psi_0 + \psi_7) - \psi_1 \otimes \psi_7 - \psi_2 \otimes \psi_7, \\
\end{split}\end{equation*}
\begin{equation}\begin{split}\label{A}
& Q_{i}^{(A_1A_2)} =  Q_{i}^{(A_1)} \otimes (\psi_0 + \psi_7), \ \ (1 \leq i \leq 5) \\
& Q_6^{(A_1A_2)} =  Q_4^{(A_1)} \otimes (\psi_0 + \psi_7), \\
& Q_7^{(A_1A_2)} =  Q_5^{(A_1)} \otimes (\psi_0 + \psi_7); \\
\end{split}\end{equation}

\noindent For $B|CA$,
\begin{equation*}\begin{split}
Y^{(B_1B_2)} = & Y^{(B_1)} \otimes (\psi_0 + \psi_7) - \psi_0 \otimes \psi_7 - \psi_2 \otimes \psi_7,   \\
\end{split}\end{equation*}
\begin{equation}\begin{split}
& Q_{i}^{(B_1B_2)} =  Q_{i}^{(B_1)} \otimes (\psi_0 + \psi_7), \ \ (1 \leq i \leq 5) \\
& Q_6^{(B_1B_2)} =  Q_4^{(B_1)} \otimes (\psi_0 + \psi_7), \\
& Q_7^{(B_1B_2)} =  Q_5^{(B_1)} \otimes (\psi_0 + \psi_7); \\
\end{split}\end{equation}

\noindent For $C|AB$,
\begin{equation*}\begin{split}
Y^{(C_1C_2)} = & Y^{(C_1)} \otimes (\psi_0 + \psi_7) - \psi_0 \otimes \psi_7  - \psi_1 \otimes \psi_7,  \\
\end{split}\end{equation*}
\begin{equation}\begin{split}
& Q_{i}^{(C_1C_2)} =  Q_{i}^{(C_1)} \otimes (\psi_0 + \psi_7), \ \ (1 \leq i \leq 5) \\
& Q_6^{(C_1C_2)} =  Q_4^{(C_1)} \otimes (\psi_0 + \psi_7), \\
& Q_7^{(C_1C_2)} =  Q_5^{(C_1)} \otimes (\psi_0 + \psi_7). \\
\end{split}\end{equation}

Now the feasibility of these solutions can be checked inductively from the feasibility of (\ref{s51}). Taking $A|BC$ for example, it's routine to check from (\ref{s52})  that
$$Y^{(A_1)} - (\psi_1 + \psi_2) > 0 = \mathrm{T}_A(Q_{2}^{(A_1)}) = \mathrm{T}_A(Q_{3}^{(A_1)}),$$
and
$$Y^{(A_1)} - (\psi_1 + \psi_2) \geq \psi_{k-1} + \mathrm{T}_A(Q_{k}^{(A_1)}) \ \ (k = 1, 4, 5).$$
Therefore, we have
\begin{equation*}\begin{split}
  & Y^{(A_1A_2)} - (\psi_1 \otimes \psi_0 + \psi_2 \otimes \psi_0)\\
= \ & Y^{(A_1)} \otimes (\psi_0 + \psi_7) - (\psi_1 + \psi_2) \otimes (\psi_0 + \psi_7) \\
> \ & 0 = \mathrm{T}_A(Q_{2}^{(A_1A_2)}) = \mathrm{T}_A(Q_{3}^{(A_1A_2)}),
\end{split}\end{equation*}
and
\begin{equation*}\begin{split}
  & Y^{(A_1A_2)} - (\psi_1 \otimes \psi_0 + \psi_2 \otimes \psi_0) \\
= \ & Y^{(A_1)} \otimes (\psi_0 + \psi_7) - (\psi_1 + \psi_2) \otimes (\psi_0 + \psi_7) \\
> \ & [\psi_{k-1} + \mathrm{T}_A(Q_{k}^{(A_1)})] \otimes (\psi_0 + \psi_7) \\
> \ & \psi_{k-1} \otimes \psi_0 + \mathrm{T}_A(Q_{k}^{(A_1A_2)})\\
&\left(\text{also} \ \ \psi_{k-1} \otimes \psi_7 + \mathrm{T}_A(Q_{k}^{(A_1A_2)})\right) \ \ \ \ \ \ (k = 1, 4, 5).
\end{split}\end{equation*}
This means that $Y^{(A_1A_2)} - \Psi_j \geq \mathrm{T}_A(Q_{j}^{(A_1A_2)})$ for all $1 \leq j \leq 7$, where $\{\Psi_1, \cdots, \Psi_7\} = \mathcal{S}_{7}$. Hence, the solution (\ref{A}) for bipartition $A|BC$ is feasibile and so forth for the other two bipartitions. Since the corresponding target values for these solutions are all smaller than $1$ ($6/7, 6/7, 6/7$ respectively), the set $\mathcal{S}_{7}$  is genuinely nonlocal.
\end{proof}

\bigskip

{\noindent \bf Proof of Lemma \ref{Dualsplit}: }

A general standard form of semidefinite program is defined as:
\begin{equation*}\begin{split}
\alpha & =  \ \max_{X} \ Tr(MX) \\
                        s.  t. \  & \ \Omega(X)= N, \\
                             & \ X \geq 0 \\
\end{split}\end{equation*}
where $M$, $N$ and $X$ (called optimization variable) are Hermitian operators and $\Omega(\cdot)$ is a linear map that preserves Hermiticity. The problem is called primal problem and its dual problem is defined as:
\begin{equation*}\begin{split}
\beta & =  \ \min_{Y} \ Tr(NY) \\
                        s.  t. \  & \ \Omega^{\ast}(Y) \leq M \\
\end{split}\end{equation*}
where the optimization variable $Y$ is also Hermitian and $\Omega^{\ast}$ is the adjoint map of $\Omega$ satisfying $Tr[\Omega(X) \cdot Y] = Tr[X \cdot \Omega^{\ast}(Y)]$ for arbitrary Hermitian $X$ and $Y$.

Putting the semidefinite program (\ref{SDPsplit}) into this standard form, we have that

\begin{widetext}
\begin{equation*}
M = \arraycolsep = 0.33em \left(\begin{array}{c} \begin{matrix}
& \Phi_1 & & & & & &  \\
& &\ddots & & & & 0's &  \\
& & & \Phi_s & & & &  \\
& & & & 0  & & &  \\
& & & & & \ddots & & \\
& & 0's & & & & \ddots &  \\
& & & & & & & 0 \\
\end{matrix}\end{array}\right), \ \ \ \ \ \ \ N = \arraycolsep = 0.33em
\left(\begin{array}{c} \begin{matrix}
& I_{ABC} & & & 0's & \\
& & 0  & & & \\
& & & \ddots & & \\
& 0's & & & \ddots &  \\
& & & & & 0  \\
\end{matrix}\end{array}\right),
\end{equation*}
and
\begin{equation*}
\Omega(X) = \arraycolsep = 0.05em \left(\begin{array}{c} \begin{matrix}
 X_1 + \cdots + X_s & & & & & & & & &  \\
 & \mathrm{T}_A(X_1)&-X^{\prime}_1 & & & & & 0's & &\\
 & & \ddots & & & & & & & \\
 & & \mathrm{T}_A(&X_s)-X^{\prime}_s & & & & & & \\
 & & & & \mathrm{T}_B(X_1)&-X^{\prime}_{s+1} & & & &\\
 & & & & & \ddots & & & & \\
 & & & & & \mathrm{T}_B(X_s&)-X^{\prime}_{2s} & & & \\
 & & & & & & & \mathrm{T}_C(X_1)&-X^{\prime}_{2s+1} & \\
 & & 0's & & & & & & \ddots & \\
 & & & & & & & & \mathrm{T}_C(X_s)&-X^{\prime}_{3s} \\
\end{matrix}\end{array}\right)
\end{equation*}
where the optimization variable $X \geq 0$ has been written as
\begin{equation*}
X = \frac1s \arraycolsep = 0.25em \left(\begin{array}{c} \begin{matrix}
& X_1 & & & & & &  \\
& & \ddots & & & & \ast & \\
& & & X_s & & & &  \\
& & & & X^{\prime}_1 & & &  \\
& & & & & \ddots & & \\
& & \ast & & & & \ddots & \\
& & & & & & & X^{\prime}_{3s} \\
\end{matrix}\end{array}\right).
\end{equation*}
Namely, $X_1, \cdots, X_s$, $X^{\prime}_1, \cdots, X^{\prime}_{3s}$ are the diagonal blocks of $sX$ and they are therefore semidefinite positive. For its dual problem, denote
\begin{equation*}
Y = \frac1s \arraycolsep = 0.25em \left(\begin{array}{c} \begin{matrix}
Y^{\prime} & & & &  \\
 & -Y_1 & & \ast &  \\
 & & \ddots & & \\
 & \ast & & \ddots & \\
 & & & & -Y_{3s} \\
\end{matrix}\end{array}\right)
\end{equation*}
the optimization variable and we have
\begin{equation*}\begin{split}
 & \ Tr[\Omega(X) \cdot Y] \\
= & \ \frac1s Tr[(X_1 + \cdots + X_s) \cdot Y^{\prime}] - \frac1s \sum_{i=1}^{s} Tr\{ [ \mathrm{T}_A(X_i)-X^{\prime}_i ] \cdot Y_i \} - \frac1s \sum_{i=1}^{s} Tr\{ [ \mathrm{T}_B(X_i)-X^{\prime}_{s+i} ] \cdot Y_{s+i} \}\\
 & - \frac1s \sum_{i=1}^{s} Tr\{ [ \mathrm{T}_C(X_i)-X^{\prime}_{2s+i} ] \cdot Y_{2s+i} \} \\
= & \ Tr\left[ \frac1s \arraycolsep = 0.05em \left(\begin{array}{c} \begin{matrix}
& X_1 & & & & & &  \\
& & \ddots & & & & \ast & \\
& & & X_s & & & &  \\
& & & & X^{\prime}_1 & & &  \\
& & & & & \ddots & & \\
& & \ast & & & & \ddots & \\
& & & & & & & X^{\prime}_{3s} \\
\end{matrix}\end{array}\right)_{=X}  \cdot \arraycolsep = 0.02em \ \left(\begin{array}{c} \begin{matrix}
& Y^{\prime}-\mathrm{T}_A(Y_1) \ -&\mathrm{T}_B(Y_{s+1})-\mathrm{T}_C(Y_{2s+1}) & & & & 0's &  \\
& & \ddots & & & & & & \\
& & Y^{\prime}-\mathrm{T}_A(Y_s)-\mathrm{T}_B(Y_{2s})-&\mathrm{T}_C(Y_{3s}) & & & &  \\
& & & & \ & Y_1 & & &  \\
& & & & & & \ddots & & \\
& 0's & & & & & & \ddots & \\
& & & & & & & & Y_{3s} \\
\end{matrix}\end{array}\right)_{=\Omega^{\ast}(Y)}
\right],
\end{split}\end{equation*}
where we apply $Tr[\mathrm{T}_{A}(V)\cdot W] = Tr[V \cdot \mathrm{T}_{A}(W)]$ (also for $\mathrm{T}_{B}$, $\mathrm{T}_{C}$) in the second equality. Consequently, by the constraint $\Omega^{\ast}(Y) \geq M$ we get the dual problem as:
\begin{equation*}\begin{split}
& \beta= \min_{_{Y^{\prime},   Y_1, \cdots,  Y_{3s}}} \frac1s \ Tr(Y^{\prime}) \\
                        s. t. \ & \ Y^{\prime} - \Phi_i \geq \mathrm{T}_A (Y_i) + \mathrm{T}_B (Y_{s+i}) + \mathrm{T}_C (Y_{2s+i}), \ \ (1 \leq i \leq s)\\
                                  & \ Y_1, \cdots, Y_{3s} \geq 0
\end{split}\end{equation*}
which is equivalent to problem (\ref{dualsplit}).
\end{widetext}

\end{document}